\documentclass[11pt,oneside,english]{amsart}
\usepackage[T1]{fontenc}
\usepackage[latin9]{inputenc}
\usepackage[letterpaper]{geometry}
\usepackage{amsthm}
\usepackage{amstext}
\usepackage{amssymb}
\usepackage{esint}

\makeatletter

\let\SF@@footnote\footnote
\def\footnote{\ifx\protect\@typeset@protect
    \expandafter\SF@@footnote
  \else
    \expandafter\SF@gobble@opt
  \fi
}
\expandafter\def\csname SF@gobble@opt \endcsname{\@ifnextchar[
  \SF@gobble@twobracket
  \@gobble
}
\edef\SF@gobble@opt{\noexpand\protect
  \expandafter\noexpand\csname SF@gobble@opt \endcsname}
\def\SF@gobble@twobracket[#1]#2{}

\numberwithin{equation}{section}
\numberwithin{figure}{section}
\theoremstyle{plain}
\newtheorem{thm}{Theorem}[section]
  \theoremstyle{plain}
  \newtheorem{cor}[thm]{Corollary}
  \theoremstyle{plain}
  \newtheorem{prop}[thm]{Proposition}
  \theoremstyle{remark}
  \newtheorem{rem}[thm]{Remark}

\def\makebbb#1{
    \expandafter\gdef\csname#1\endcsname{
        \ensuremath{\Bbb{#1}}}
}\makebbb{R}\makebbb{N}\makebbb{Z}\makebbb{C}\makebbb{H}\makebbb{E}\makebbb{Q}\makebbb{P}\makebbb{B}\makebbb{K}\makebbb{E}

\usepackage{babel}

\usepackage{babel}

\makeatother

\usepackage{babel}

\makeatother

\usepackage{babel}

\begin{document}

\title{Sharp asymptotics for Toeplitz determinants and convergence towards
the gaussian free field on Riemann surfaces}

\author{Robert J. Berman}

\address{Chalmers University of Technology and University of Gothenburg}

\email{robertb@chalmers.se}
\begin{abstract}
We consider canonical determinantal random point processes with $N$
particles on a compact Riemann surface $X$ defined with respect to
the constant curvature metric. We establish strong exponential concentration
of measure type properties involving Dirichlet norms of linear statistics.
This gives an optimal Central Limit Theorem (CLT), saying that the
fluctuations of the corresponding empirical measures converge, in
the large $N$ limt, towards the Laplacian of the Gaussian free field
on $X$ in the strongest possible sense. The CLT is also shown to
be equivalent to a new sharp strong Szegö type theorem for Toeplitz
determinants in this context. One of the ingredients in the proofs
are new Bergman kernel asymptotics providing exponentially small error
terms in a constant curvature setting.

\tableofcontents{}
\end{abstract}
\maketitle

\section{Introduction}

This paper is one in a series which deal with $N-$particle determinantal
point processes on a polarized compact complex manifold $X,$ i.e.
associated to high powers of an ample line bundle $L\rightarrow X.$
In the paper in \cite{berm6} a general Large Deviation Principle
(LDP) for such processes was established in the large $N-$limit showing
that the empirical measures converge exponentially towards the deterministic
pluripotential equilibrium measure. Moreover, in the paper \cite{berm2}
a Central Limit Theorem (CLT) was obtained, showing that the fluctuations
in the {}``bulk'' may be desribed by a Gaussian free field in the
case of \emph{smooth }test functions (linear statistics). In the present
paper we specialize to the lowest dimensinal case when $X$ is a Riemann
surface and the corresponding $N-$particle point processes are the
{}``canonical'' ones, i.e. the they are induced by the Kähler-Einstein
metric on $X.$ In this setting we obtain sharp versions of the upper
large deviation bound and show that the convergence towards the Gaussian
free field holds in the strongest possible sense, i.e. for linear
statistics with minimal regularity assumptions (finite Dirichlet norm).
This CLT is equivalent to a new sharp strong Szegö type theorem for
Toeplitz determinants in this context. The results are obtained from
new {}``determinantal'' Moser-Trudinger type inequalities, which
imply strong concentration of measures properties. The proof of these
latter inequalities is based on a convexity argument in the space
of all Kähler metrics, combined with Bergman kernel asymptotics and
potential theory.

\subsection{The general setup}

Let $L\rightarrow X$ be an ample holomorphic line bundle over a compact
complex manifold $X$ of dimension $n.$ We will denote by $H^{0}(X,L)$
the $N-$dimensional vector space of all global holomorphic sections
of $L.$ Given the geometric data $(\nu,\left\Vert \cdot\right\Vert )$
consisting of a probability measure $\nu$ on $X$ and a continuous
Hermitian metric $\left\Vert \cdot\right\Vert $ on $L$ one obtains
an associated probability measure $\mu^{(N)}$ on the $N-$fold product
$X^{N}$ defined as \begin{equation}
\mu^{(N)}:=\frac{1}{\mathcal{Z}_{N}}\left\Vert \det\Psi\right\Vert ^{2}(x_{1},...x_{N})\nu(x_{1})\otimes\cdots\otimes\nu(x_{N})\label{eq:def of det process intro}\end{equation}
where $\det\Psi$ is a holomorphic section of the pulled-back line
bundle $L^{\boxtimes N}$ over $X^{N}$ representing the $N$th (i.e.
maximal) exterior power of $H^{0}(X,L)$ and $\mathcal{Z}_{N}$ is
the normalizing constant. Concretely, fixing a base $(\Psi_{i})_{i=1}^{N}$
in $H^{0}(X,L)$ we can take \begin{equation}
(\det\Psi)(x_{1},...,x_{N_{k}})=\det(\Psi_{i}(x_{j}))\label{eq:slater det}\end{equation}
We will denote $\frac{i}{2\pi}$ times the curvature two-form of the
metric on $L$ by $\omega$ (compared with mathematical physics notation
$\omega=\frac{i}{2\pi}F_{A}$ where $A$ is the Chern connection induced
by the metric on $L$). It will be convenient to take the pair $(\omega,\nu),$
which will refer to as a\emph{ weighted measure}, as the given geometric
data. The\emph{ empirical measure} of the ensemble above is the following
random measure: \begin{equation}
(x_{1},...,x_{N})\mapsto\delta_{N}:=\sum_{i=1}^{N}\delta_{x_{i}}\label{eq:intro random measure}\end{equation}
which associates to any $N-$particle configuration $(x_{1},...,x_{N})$
the sum of the delta measures on the corresponding points in $X.$
In probabilistic terms this setting hence defines a \emph{determinantal
random point process on $X$ with $N$ particles} \cite{h-k-p,j2}. 

If the correponding $L^{2}-$norm on $H^{0}(X,L)$ \begin{equation}
\left\Vert \Psi\right\Vert _{X}^{2}=\left\langle \Psi,\Psi\right\rangle _{X}:=\int_{X}\left\Vert \Psi(x)\right\Vert ^{2}d\nu(x)\label{eq:general l2 norm intro}\end{equation}
is non-degenerate (which will always be the case in this paper) then
the probability measure $\mu^{(N)}$ on $X^{N}$ may be expressed
as a determinant of the\emph{ Bergman kernel} of the Hilbert space
$(H^{0}(X,L),\left\Vert \cdot\right\Vert _{X}),$ i.e. the integral
kernel of the corresponding orthogonal projection $\Pi.$ A central
role in this paper will be played by the\emph{ logarithmic generating
function} (or \emph{free enegy}) \[
\log\E(e^{-(\sum_{i=1}^{N}(\phi(x_{i})})\]
of the \emph{linear statistic} \begin{equation}
\sum_{i=1}^{N}\phi(x_{i}),\label{eq:linear stat}\end{equation}
where $\E$ denotes the expectation wrt the ensemble $(X^{N},\mu^{(N)}),$
i.e. $\E(\cdot)=\int_{X^{N}}(\cdot)\mu^{(N)}.$ By a well-known formula
going back to the work of Heine in the theory of orthogonal polynomials
the expectation above can also be writen as a \emph{Toeplitz determinant}
with symbol $e^{-\phi}:$ \begin{equation}
\E(e^{-(\sum_{i=1}^{N}(\phi(x_{i}))})=\det(\left\langle e^{-\phi}\Psi_{i},\Psi_{j}\right\rangle _{X})\left(=\det T[e^{-\phi}]\right)\label{eq:log moment as toeplitz det}\end{equation}
 where $(\Psi_{i})_{i=1}^{N}$ is an orthonormal base in the Hilbert
space $(H^{0}(X,L),\left\Vert \cdot\right\Vert _{X})$ (and $T[f]:=\Pi(f\cdot)$
is the corresponding Toeplitz operator on $H^{0}(X,L)$ with symbol
$f).$ Replacing $L$ with its $k$ th tensor power, which we will
write in additive notation as $kL,$ yields,a sequence of point processes
on $X$ of an increasing number $N_{k}$ of particles. We will be
concerned with the asymptotic situation when $k\rightarrow\infty.$
This corresponds to a large $N-$limit of many particles, since \[
N_{k}:=\dim H^{0}(X,kL)=Vk^{n}+o(k^{n})\]
where the constant $V$ is, by definition, the \emph{volume} of $L.$

As shown in \cite{berm6} the normalized empirical measure $\delta_{N}/N_{k}$
converges towards a pluripotential equlibrium measure $\mu_{eq},$
exponentially in probability. In particular, letting

\begin{equation}
\epsilon_{N_{k},\lambda}(\phi):=\mbox{Prob}\left\{ \left|\frac{1}{N_{k}}(\phi(x_{1})+....+\phi(x_{N_{k}}))-\int_{X}\mu_{eq}\phi\right|>\lambda\right\} \label{eq:def of tail}\end{equation}
denote the\emph{ tail of the linear statistic} determined by $\phi,$
at level $k,$ it was shown that $\epsilon_{N_{k},\lambda}(\phi)\rightarrow0$
as $k\rightarrow\infty$ for any $\lambda>0$ at a rate of the order
$e^{-k^{n+1}/C}.$ In the case when $X$ is a Riemann surface the
curvature current $\omega$ of the metric on $L$ is semi-positive
(so that $\mu_{eq}=\omega)$ the following more precise estimate was
obtained: \begin{equation}
\epsilon_{N_{k},\lambda}(\phi)\leq2\exp(-N_{k}^{2}\left(\frac{2V\lambda^{2}}{\left\Vert d\phi\right\Vert _{X}^{2}}(1+o(1))\right))\label{eq:exp conc general riemann s}\end{equation}
 where the error term $o(1)$ denotes a sequence tending to zero as
$k\rightarrow\infty$ (but depending on $\phi$).

\subsection{The canonical setting on a Riemann surface }

Let now $L\rightarrow X$ be a line bundle of positive volume (degree)
$V$ over a Riemann surface $X$ of genus $g.$ It determines a particular
sequence of determinantal point process that we will refer to as the
\emph{canonical deteterminantal point process on $X$ associated to
$kL.$ }These processes are obtained by letting $(\nu,\omega)=(\omega/V,\omega)$
for $\omega$ the the unique volume form on $X$ of volume $V$ such
that Riemannian metric determined by $\omega$ has constant scalar
curvature. By the Riemann-Roch theorem we have (for $k$ sufficently
large) \begin{equation}
N_{k}=kV-(g-1)\label{eq:r-r relation intro}\end{equation}
giving a simple relation between the level $k$ and the corresponding
number of particles $N_{k}$. Accordingly, it will be convenient to
talk about the \emph{canonical determinantal random point process
on $X$ with $N$ particles }and use $N$ as the asymptotic parameter.\emph{
}Strictly speaking $N(=N_{k})$ only determines $L$ up to twisting
by a flat line bundle, but the results will be independant of the
flat line bundle. Physically, the canonical processes associated to
$kL$ represents the groundstate of a gas of spin-polarized free ferrmions
in the {}``uniform'' magnetic field $kF_{A}$ where $\omega=\frac{i}{2\pi}F_{A}$
and $A$ is a unitary connection on $L$ (see \cite{berm6} and references
therein). Equivalently, these processes are defined by the lowest
Landau level of the corresponding magnetic Schrödinger operator.

The simplest case of this setting occurs when $g=0,$ i.e. $X$ is
the Riemann sphere and then $H^{0}(X,kL)$ may be identified with
the space of all polynomials on the affine piece $\C$ of degree at
most $k=N-1$ equipped with the usual $SU(2)-$invariant Hermitian
product. Alternatively, embedding $X$ as the unit-sphere in Euclidian
$\R^{3}$ the $N-$point correlation function of the process, i.e.
the density of the probability measure, may be explicitely expanded
as \[
\rho^{(N)}(x_{1},...,x_{N}):=\Pi_{1\leq i<j\leq N}\left\Vert x_{i}-x_{j}\right\Vert ^{2}/Z_{N}\]
where $1/Z_{N}=N^{N}\binom{N-1}{0}...\binom{N-1}{N-1}/N!.$ In the
physics litterature this ensemble also appears as a \emph{Coulomb
gas} of $N$ unit-charge particles (i.e a one component plasma) confined
to the sphere in a neutralizing uniform background $\omega$ (see
for example \cite{ca}). An interesting random matrix model for this
process was recently given in \cite{kr}. In the higher genus case
the role of polynomials are played by theta functions and modular
(automorphic) forms on the universal covers $\C$ and $\H$ of $X$
(when $g=1$ and $g>1$ respectively) equipped with their standard
Hermitian products. See for example \cite{f1} for the case $g=1$
in connection to fermions and bosonization. When $g>1$ the Riemann
surface $X$ may be represented as the quotient $\Gamma/\H$ of the
upper half-plane with a suitable discrete subgroup $\Gamma$ of $SL(2,\R).$
Taking $L:=\frac{1}{2}K_{X},$ where $K_{X}$ denotes the canonical
line bundle $K_{X}=T^{*}X$ (using the induced spin structure to take
the square root of $K_{X}$) realizes $H^{0}(X,kL)$ as the Hilbert
space of all modular forms of weight $k,$ i.e. all holomorphic funtions
on $\H$ satisfying $f((az+b)/(cz+d))=(cz+d)^{k}f(z)$ equipped with
the Petterson norm \[
\left\Vert f\right\Vert _{X}^{2}:=\int_{\Gamma/\H}|f|^{2}y^{k}\frac{dx\wedge dy}{y^{2}},\]
 integrating over a fundamental domain for $\Gamma.$ In special arithmetic
situation the base $(\Psi_{i})$ in \ref{eq:slater det} may be represented
by Hecke eigenfunctions (but note that we have assumed that $X$ is
smooth and compact and in particular there are no cusps)\cite{la}.

\subsection{Statement of the main results}

It will be convenient to use the following conformally invariant notation
for the normalized Dirichlet norm of a function $\phi$ on $X,$ i.e.
the $L^{2}-$norm of its gradient times $1/4\pi:$ \[
\left\Vert d\phi\right\Vert _{X}^{2}:=\int_{X}d\phi\wedge d^{c}\phi:\left(=\frac{i}{2\pi}\int_{X}\partial\phi\wedge\bar{\partial}\phi\right)\]
We will obtain a very useful Moser-Trudinger type inequality for the
canonical determinantal point processes, which generalizes Onofri's
sharp version of the Moser-Trudinger inequality \cite{on} (obtained
when $X$ is the two-sphere and $N=1).$ 
\begin{thm}
\label{thm:(determinatal-moser...)...For-the}Let $X$ be a genus
$g$ Riemann surface and consider the canonical determinantal point
process on $X$ with $N$ particles. It satisfies the following Moser-Trudinger
type inequality: \textup{\begin{equation}
\log\E(e^{-(\sum_{i=1}^{N}(\phi(x_{i})-\int_{X}\phi\frac{\omega}{V}))})\leq\left(\frac{1}{1+(1-g)/N}+\epsilon_{N})\right)\frac{1}{2}\left\Vert d\phi\right\Vert _{X}^{2}+\epsilon_{N}\label{eq:thm det ineq}\end{equation}
where the error term $\epsilon_{N}$ is exponentially small, i.e}.,
$\epsilon_{N}\leq Ce^{-N\delta}$ \textup{for some postive number
$C$ and $\delta$ independent of $\phi$ and where $\delta$ can
be explicitely expressed in terms of the injectivity radius of $(X,\omega)$
(see formula \ref{eq:form for delta in prop} in Prop \ref{pro:bergm as}).
Similarly, \begin{equation}
\log\E(e^{-(\sum_{i=1}^{N}(\phi(x_{i})-\E(\phi(x_{i}))})\leq\left(\frac{1}{1+(1-g)/N}+\epsilon_{N}\right)\frac{1}{2}\left\Vert d\phi\right\Vert _{X}^{2}+\epsilon_{N}\left\Vert \phi\right\Vert _{L^{1}(X)/\R}+\epsilon_{N}\label{eq:det moser intro fluct version}\end{equation}
Moreover, when $X$ is the Riemann sphere (i.e \@.$g=0)$ all the
error terms above vanish identically.}
\end{thm}
An important ingredient in the previous proof is a convexity result
of Berndtsson \cite{bern2} which in this particular case essentially
amounts to the positivity of a certain determinant line bundle over
the space of all Kähler metrics in the first Chern class of $L.$
The error terms $\epsilon_{N}$ above come from the error terms in
the Yau-Tian-Zelditch-Catlin expansion \cite{z,berm0,l-t,liu-l} for
the underlying Bergman kernel. As follows from Theorem \ref{thm:bergman kernel as}
below these error terms are exponentially small, slightly refining
previous recent results in \cite{liu,liu-l} (see section \ref{sec:Bergman-kernel-asymptotics}
 for precise formulations). 

As a simple consequence of the previous theorem we then obtain a sharp
version of the tail estimate \ref{eq:exp conc general riemann s}
for such canonical processes. The main point is that it shows that
the error term $o(1)$ appearing in the estimate \ref{eq:exp conc general riemann s}
can be taken to be independent of the function $\phi.$ As a consequence
the estimate holds with minimal regularity assumptions on $\phi:$ 
\begin{cor}
\label{cor:sharp tail intro}Let $X$ be a genus $g$ Riemann surface
and consider the canonical determinantal point process on $X$ with
$N$ particles.\textup{ Let $\phi$ be a function on $X$ such that
its differential $d\phi$ is in $L^{2}(X).$ Then the linear statistic
defined by $\phi$ has an exponentially decaying tail: \[
\epsilon_{N,\lambda}(\phi)\leq2\exp(-N^{2}\left(\frac{2\lambda^{2}}{\left\Vert d\phi\right\Vert _{X}^{2}(1+\frac{(1-g)}{N})+\epsilon_{N})}+\epsilon_{N}\right))\]
where the error terms $\epsilon_{N}$ are as in the previous theorem. }
\end{cor}
We will also show that the Moser-Trudinger inequality in Theorem \ref{thm:(determinatal-moser...)...For-the}
is in fact an asymptotic \emph{equality} in the following sense:
\begin{thm}
\label{thm: szeg=0000F6type}(strong Szegö type theorem). Let $X$
be a genus $g$ Riemann surface and consider the canonical determinantal
point process on $X$ with $N$ particles. Let $\phi$ be a complex
valued function on $X$ such that \textup{its differential is in $L^{2}(X,\C),$
i.e. $\phi$ has finite Dirichlet norm. Then} \textup{\[
\log\E(e^{-(\sum_{i=1}^{N}(\phi(x_{i})-\int_{X}\phi\omega))})\rightarrow\frac{1}{2}\int_{X}d\phi\wedge d^{c}\phi\]
 as $N\rightarrow\infty$ and the same convergence holds when the
exponent above is replaced with the fluctuation of the linear statistic
of $\phi.$ }
\end{thm}
In \cite{berm2} it was shown that, as long as $\omega>0$ and $\phi$
is\emph{ smooth} an analogue of the convergence above holds in any
dimension $n$ if the conformally invariant norm above is replaced
by the Dirichlet norm wrt $\omega.$ But it should be emphasized that
when $n>1$ the convergence does \emph{not} hold if one relaxes the
smoothness assumption on $\phi$ to allowing a gradient in $L^{2}$
(see section \ref{sub:A-brief-acount} for counter examples). 

The previous theorem may be equivalently formulated as the following
Central Limit Theorem (CLT), valid under minimal regularity assumptions:
\begin{cor}
\textup{(CLT) The}\emph{ fluctuations} $\delta_{N}-\E(\delta_{N})$
\emph{of the empirical measure $\delta_{N}$ converge in distribution
to the the Laplacian (or rather $dd^{c})$ of the Gaussian free field
(GFF). }\textup{In other words, for any $\phi\in L^{1}(X)$ with $d\phi\in L^{2}(X)$
the fluctuations}\emph{ }\[
\sum_{i=1}^{N}(\phi(x_{i})-\E(\phi(x_{i}))\]
 \emph{of the corresponding linear statistics converge in distribution
to a centered normal random variable with variance $\left\Vert d\phi\right\Vert _{X}^{2}.$} 
\end{cor}
The GFF is also called the \emph{massless bosonic free field }in the
physics litterature. Heuristically, this is a random function wrt
the Gaussian measure on the Hilbert space of all $\phi$ (mod $\R)$
equipped with the Dirichlet norm $\left\Vert d\phi\right\Vert _{X}^{2}/2.$
For the precise definition of the GFF and its Laplacian see \cite{sh}
(Prop 2.13 and Remark 2.14) and for a comparison with the physics
litterature on Coulomb gases see section 1.3 in \cite{s-s}.

\subsection{Relations to previous results}

\subsubsection*{Exponential concentration }

A determinantal Moser-Trudinger (M-T) inequality on $S^{2},$ but
with non-optimal constants was first obtained by Fang \cite{f} building
on previous work by Gillet-Soulé concerning the $S^{1}-$invariant
case \cite{g-s}, which in turn used the classical Moser-Truding (one-particle)
inequality. The motivation came from arithmetic (Arakelov) geometry
and spectral geometry. The optimal constants on $S^{2}$ were obtained
by the author in \cite{berm5} using methods further developed in
the present paper. It would be interesting to know for which other
(determinantal) random point processes similar inequalities hold,
i.e. upper bounds on the logarithmic moment generating function of
the linear statistic defined by $\phi(x)$ in terms of the Dirichlet
norm $\left\Vert d\phi\right\Vert _{X}^{2}.$ The only previously
known case seems to be the case when the measure measure $\nu$ is
the invariant measure on $S^{1}$ (and $\omega=0),$ corresponding
to the standard unitary random matrix ensemble. Then the corresponding
inequalities follow from a simple monotonicity argument going back
to the classical work of Szegö (see for example \cite{j00} and references
therein). Recently, several works have been concerned with a weaker
form of such moment inequalities where the role of the Dirichlet norm
is played by the Lipschitz norm. These inequalities fit into a circle
of ideas sourrounding the {}``concentration of measure phenomena''
in high dimensions. We refer to the survey \cite{gu} and the book
\cite{le} for precise references. Formulated in the present settings
these latter inequalities hold for $\nu=1_{\R}e^{-v(x)}dx$ with $v(x)$
strictly convex (satisfying $d^{2}v/d^{2}x>C).$ As explained in \cite{gu},
by the Bakry-Emery theorem and Klein's lemma, the corresponding point
processes satisfy a log Sobolev inequality, which by Herbst's argument
yields the desired inequality on the logarithmic moment generating
function

\subsubsection*{Szegö type limits and CLT:s}

The convergence in Theorem \ref{thm: szeg=0000F6type} (and its Corollary)
in the case when $X=S^{2}$ was first obtained by Ryder-Virag \cite{r-v2},
using combinatorial (and diagrammatic) arguments to estimate the cumulants
(i.e. the coefficients in the Taylor expansion of the logarithmic
moment generating function), combined with estimates on the 2-point
functions. They also obtained analagous results for the homogenous
determinantal point processes on the other two simply connected Riemann
surfaces, i.e on $\C$ and $\H.$ However, in the latter cases the
processes have an infinite number of particles and are hence different
from the sequence of non-homogenous ones considered in the present
paper on a \emph{compact} Riemann surfaces of genus $g>0.$ In the
circle case (refered to above), assuming $\phi$\emph{ smooth, }the
analogue of the convergence in Thm \ref{thm: szeg=0000F6type} is
the celebrated Szegö strong limit theorem from 1952. In this case
the Dirichlet norm of $\phi$ has to be replaced by the Dirichlet
norm of the harmonic extension of $\phi$ to the unit-disc. The result
of Szegö was motivated by Onsager's work on phase transitions for
the 2D Ising model. The case of a general $\phi$ was eventually shown
by Ibragimov \cite{ib}. A new proof was then given by Kurt Johansson
\cite{j00}, who also pointed out the relation to a CLT for the unitary
random matrix ensemble. See also \cite{d-e} for generalizations of
the latter CLT using explicit moment calculations and harmonic analysis.
We refer to the survey \cite{sim} for an interesting account of the
history of Szegö's theorem. It is also interesting to compare the
appearence of exponentially small error terms in the inequaties \ref{thm:(determinatal-moser...)...For-the}
with the exponentially small error terms obtained in \cite{sim} in
the context of the classical strong Szegö theorem. The proof in the
Riemann surface cases in the present paper is partly inspired by the
argument in \cite{j00}, where the determinantal Moser-Trudinger inequalities
on $S^{1}$ (refered to above) were used to reduce the\emph{ upper}
bound in the convergence to the smooth case, also using analytic continuation.
There are also similar convergence results for other weighted measures
in the plane appearing in Random Matrix Theory, but regularity assumptions
on $\phi$ are then imposed \cite{j00,j0}. It should be emphasizes
that the classical strong Szegö theorem has previously been extended
and extensively studied in various other directions, notably in the
context of pseudo-differential operators and in particular Schrödinger
operators (see \cite{sim} and references therein). Compared to the
present paper the role of Schrödinger operators is here played by
\emph{magnetic} Schrödinger operators. Finally, it may also be interesting
to compare the CLT above with the central limit theorem and variance
asymptotics obtained very recently in \cite{n-s} for non-smooth linear
statistics in the different context of random point processes defined
by zeroes of Gaussian entire functions.

\subsection*{Acknowledgment}

The author is grateful to Gerard Freixas i Montplet, Balint Virag,
Manjunath Krishnapur and Steve Zelditch for helpful comments and their
interest in this work.

\subsection{\label{sub:Notation-and-general}Notation%
\footnote{general references for this section are the books \cite{gr-ha,de4}.
See also \cite{b-v-} for the Riemann surface case.%
}}

Let $L\rightarrow X$ be a holomorphic line bundle over a compact
complex manifold $X.$

\subsubsection{Metrics on $L$}

We will fix, once and for all, a Hermitian metric $\left\Vert \cdot\right\Vert $
on $L.$ Its curvature form times the normalization factor $\frac{i}{2\pi}$
will be denoted by $\omega.$ The normalization is made so that $[\omega]$
defines an\emph{ integer} cohomology class, i.e. $[\omega]\in H^{2}(X,\Z).$
The local description of $\left\Vert \cdot\right\Vert $ is as follows:
let $s$ be a trivializing local holomorphic section of $L,$ i.e.
$s$ is non-vanishing an a given open set $U$ in $X.$ Then we define
the local\emph{ weight }$\Phi$ of the metric $\left\Vert \cdot\right\Vert $
by the relation \[
\left\Vert s\right\Vert ^{2}=e^{-\Phi}\]
The (normalized) curvature current $\omega$ may now by defined by
the following expression: \[
\omega=\frac{i}{2\pi}\partial\overline{\partial}\Phi:=dd^{c}\Phi,\]
 (where we, as usual, have introduced the real operator $d^{c}:=i(-\partial+\overline{\partial})/4\pi$
to absorb the factor $\frac{i}{2\pi}).$ The point is that, even though
the function $\phi$ is merely locally well-defined the form $\omega$
is globally well-defined (as any two local weights differ by $\log|g|^{2}$
for $g$ a non-vanishing holomorphic function). The current $\omega$
is said to be\emph{ positive} if the weight $\Phi$ is \emph{plurisubharmonic
(psh). }If $\Phi$ is smooth this simply means that the\emph{ }Hermitian
matrix $\omega_{ij}=(\frac{\partial^{2}\Phi}{\partial z_{i}\partial\bar{z_{j}}})$
is positive definite (i.e. $\omega$ is a Kähler form) and in general
it means that, locally, $\Phi$ can be written as a decreasing limit
of such smooth functions.

\subsubsection{Holomorphic sections of $L$}

We will denote by $H^{0}(X,L)$ the space of all global holomorphic
sections of $L.$ In a local trivialization as above any element $\Psi$
in $H^{0}(X,L)$ may be represented by a local holomorphic function
$f,$ i.e. \[
\Psi=fs\]
 The squared point-wise norm $\left\Vert \Psi\right\Vert ^{2}(x)$
of $\Psi,$ which is a globally well-defined function on $X,$ may
hence be locally written as \[
\left\Vert \Psi\right\Vert ^{2}(x)=(|f|^{2}e^{-\Phi})(x)\]
It will be convenient to take the curvature current $\omega$ as our
geometric data associated to the line bundle $L.$ Strictly speaking,
it only determines the metric $\left\Vert \cdot\right\Vert $ up to
a multiplicative constant but all the geometric and probabilistic
constructions that we will make are independent of the constant.

\subsubsection{Metrics and weights vs $\omega-$ psh functions}

Having fixed a continuous Hermitian metric $\left\Vert \cdot\right\Vert $
on $L$ with (local) weight $\Phi_{0}$ any other metric may be written
as \[
\left\Vert \cdot\right\Vert _{\phi}^{2}:=e^{-\phi}\left\Vert \cdot\right\Vert ^{2}\]
for a continuous function $\phi$ on $X,$ i.e. $\phi\in C^{0}(X).$
In other words, the local weight of the metric $\left\Vert \cdot\right\Vert _{\phi}$
may be written as $\Phi=\phi+\Phi_{0}$ and hence its curvature current
may be written as \[
dd^{c}\Phi=\omega+dd^{c}\phi:=\omega_{\phi}\]
This means that we have a correspondence between the space of all
(singular) metrics on $L$ with positive curvature current and the
space $PSH(X,\omega)$ of all upper-semi continuous functions on $X$
such that $\omega_{\phi}\geq0$ in the sense of currents. Note for
example, that if $\Psi\in H^{0}(X,L)$ then $\log\left\Vert \Psi\right\Vert ^{2}\in PSH(X,\omega).$
In particular, in the Riemann surface case $PSH(X,\omega)(=SH(X,\omega))$
is the space of all usc functions $\phi$ such that $\Delta_{\omega}\phi\geq-1,$
where $\Delta_{\omega}$ denotes the Laplacian wrt the Riemannian
metric corresponding to $\omega,$ i.e. \[
\Delta_{\omega}\phi=(dd^{c}\phi)/\omega\]
(where by our normalizations $\Delta_{\omega}=\frac{1}{4\pi}(\frac{\partial^{2}}{\partial^{2}x}+\frac{\partial^{2}}{\partial^{2}y})$
in the case when $\omega$ is locally Euclidean).

\section{\label{sec:Exponential-concentration-(Proof}Canonical point processes
(Proofs of the main results)}

For a general Kähler manifold $(X,\omega)$ there is well-known energy
type functional which may be written as \begin{equation}
\mathcal{E}_{\omega}(\phi):=\frac{1}{(n+1)!V}\sum_{j=0}^{n}\int_{X}\omega_{\phi}^{j}\wedge(\omega)^{n-j}\label{eq:bi-energy}\end{equation}
Up to normalization it can be defined as the primitive of the Monge-Ampère
operator seen as a one-form on the space of all Kähler potentials
$\phi$ (and it was in this form it was first introduced by Mabuchi
in Kähler geometry; see \cite{berm6} and references therein). This
means that $d\mathcal{E}_{\omega}(\phi)=\omega_{\phi}^{n}/V$ (in
the sense of formula \ref{eq:first variational form} below)

We now turn to the case when $X$ is a Riemann surface, i.e. $n=1.$
In particular, after an integration by parts $\mathcal{E}_{\omega}$
can then be expressed in terms of the usual Dirichlet energy on a
Riemann surface: \begin{equation}
V\mathcal{E}_{\omega}(\phi)=-\frac{1}{2}\int d\phi\wedge d^{c}\phi+\int\phi\omega\label{eq:eomega as dirch}\end{equation}
Following \cite{berm5} it will also be convenient to consider a variant
of the setting given in the introduction of the paper where the Hilbert
space is the space $H^{0}(X,kL+K_{X})$ of holomorphic one-form with
values in $L$ equipped with the canonical Hermitian product induced
by the weight $\Phi$ on $L:$ . \begin{equation}
\left\langle \Psi,\Psi\right\rangle _{X}:=i\int_{X}\Psi\wedge\bar{\Psi}e^{-k\Phi}\label{eq:herm norm in adjoint setting}\end{equation}
(equivalently, one picks a volume form $\mu$ on $X$ and takes $1/\mu$
as the metric on $K_{X}).$ We will call this the \emph{adjoint setting
}and the corresponding process on $X$ \emph{the adjoint determinantal
point process at level $k.$ }Anyway, as explained below, the adjoint
and the canonical point processes coincide when the curvature form
$\omega$ of $\Phi$ has constant curvature.\emph{ }We note that if
$\delta_{N}$ denotes the empirical measure for the adjoint process
with $N$ particles then \[
\E(\delta_{N})=i\sum_{i=1}^{N}\Psi_{i}\wedge\bar{\Psi}_{i}e^{-\Phi},\]
for an orthonormal base $(\Psi_{i}),$ i.e. $\E(\delta_{N})$ is equal
to the restriction to the diagonal of the Bergman kernel $K_{\Phi}(x,y)$
of $H^{0}(X,kL+K_{X});$ see section \ref{sec:Bergman-kernel-asymptotics}. 
\begin{prop}
\label{pro:bergm as}Let $L\rightarrow X$ be a line bundle of degree
$V$ over a Riemann surface of genus $g.$ Assume that $L$ is equipped
with a metric $e^{-\Phi}$ with strictly positive curvature form $\omega(=dd^{c}\Phi)$
such that the Riemannian metric on $X$ defined by $\omega$ has constant
scalar curvature $R(=(2-2g/V).$ Then the canonical determinantal
point processes associated to $kL$ (with $N(=N_{k})$ particles)
satisfy \[
\sup_{X}|\frac{\E_{N}(\delta_{N}/N)}{\omega/V}-1|\leq\epsilon_{N},\]
 where $\epsilon_{N}$ is exponentially small, i.e. $\epsilon_{N}\leq Ce^{-\delta N}.$
In the case $g=0$ we have $\epsilon_{N}=0$ and when $g>0$ the constant
$\delta$ can be taken to be arbitrarily close to \begin{equation}
\frac{2}{RV}\log(\cosh(\sqrt{\frac{\pi R}{2}}I(X))\label{eq:form for delta in prop}\end{equation}
 where $I(X))$ is the injectivity radius of $X$ (which coincides
with half the length of the shortest geodesic on $X).$ \end{prop}
\begin{proof}
To simplify the notation we set $V=1$ (the case $V\neq1$ follows
from trivial scalings). First we recall that in the general setting
where \ref{eq:general l2 norm intro} defines a Hilbert norm on $H^{0}(X,kL)$
we have the basic relation $\E_{N}(\delta_{N})=B_{k}(x)d\nu$ where
$B_{k}(x)$ is the point-wise norm of the corresponding Bergman kernel
and hence Theorem \ref{thm:bergman kernel as} below (or its corollary)
gives $\E_{N}(\delta_{N})=k+R/2+\mathcal{O}(e^{-\delta k}).$ Since,
$N=\int B_{k}(x)d\nu$ it follows that $N=k+R/2$ for $k>>1$ (a special
case of the Riemann-Roch theorem) concluding the proof of the proposition. 
\end{proof}

\subsection{Proof of Theorem \ref{thm:(determinatal-moser...)...For-the} (determinantal
Moser-Trudinger inequality)}

We will start by proving the following non-asymptotic inequality.
\begin{prop}
\label{pro:moser for omega psh}Let $L\rightarrow X$ be a line bundle
over a Riemann surface equipped with a smooth metric with strictly
positive curvature form $\omega.$ Consider the corresponding adjoint
determinantal point process. Then the following estimate holds \textup{\[
\frac{1}{N}\log\E(e^{-\phi})-\mathcal{E}_{\omega}(\phi)\leq\sup_{X}|\frac{\E(\delta/N)}{\omega/V}-1|(-\mathcal{E}_{\omega}(\phi-\sup_{X}\phi))\]
 for any smooth function $\phi$ satisfying }$\omega_{\phi}:=dd^{c}\phi+\omega\geq0,$
\textup{where $\delta(=\delta_{N})$ denotes the empirical measure
of the process with $N$ particles.}
\end{prop}
The proof is a simple modification of the proof Theorem 33 in \cite{berm5}.
As a courtesy to the reader we will recall the argument in \cite{berm5}.
An important ingredient in the proof is the notion of a \emph{$C^{0}-$geodesic
}(wrt the Mabuchi metric) connecting $\phi_{0}$ and $\phi_{1}$ in
$C^{0}(X)\cap PSH(X,\omega).$ This may be defined as the continuous
path $\phi_{t}(=\phi(\cdot.t))$ connecting $\phi_{0}$ and $\phi_{1}$
in $C^{0}(X)\cap PSH(X,\omega)$ obtained as the upper envelope of
all $S^{1}-$ invariant $\pi^{*}\omega-$psh extensions to the $n+1-$dimensional
complex manifold with boundary $M:=X\times([0,1]\times S^{1})$ (where
$\pi$ denotes the projection from $X\times[0,1[\times S^{1}$ to
$X).$ In particular, $\phi_{t}$ is convex in the real paramter $t\in[0,1]$
and satisfies the homogenous Monge-Ampère equation in the interiour
of $M:$ \begin{equation}
\partial_{t}\partial_{t}\phi_{t}-|(\bar{\partial}_{X}(\partial_{t}\phi_{t})|_{\omega_{\phi_{t}}}^{2}=0\label{eq:geod eq}\end{equation}
in the weak sense of pluripotential theory (see \cite{berm5} for
the precise construction). The following variational formulae are
well-known (and straight-forward): \begin{equation}
(i)\,-\frac{1}{N}d(\log\E(e^{-\phi_{t}})/dt=\left\langle \E_{\omega_{\phi_{t}}}(\delta/N),d\phi_{t}/dt\right\rangle ,\,\,\,(ii)\, d\mathcal{E}_{\omega_{0}}(\phi_{t})/dt=\frac{1}{V}\left\langle \omega_{\phi_{t}},d\phi_{t}/dt\right\rangle \label{eq:first variational form}\end{equation}
Moreover, if $\phi_{t}$ is a $C^{0}-$geodesic in $\mbox{Psh}(X,\omega)$
then \[
(i')\,\,\log\E(e^{-\phi_{t}})\,\mbox{is\,\ concave,\,\,\,\ensuremath{(ii')\,\mathcal{E}_{\omega_{0}}(\phi_{t})\,}is\,\ affine}\]
in the real parameter $t$ (note however that $\log\E(e^{-\phi_{t}})$
is \emph{convex} along \emph{affine} curves; compare Remark \ref{rem:lower bd}
below). The item $(i')$ above follows from the Toeplitz determinant
representation \ref{eq:log moment as toeplitz det} combined with
the positivity results for direct image bundles in \cite{bern2}.
See also the appendix in \cite{berm5} for another proof of $(i')$
using the structure of determinantal point processes. The key point
is the following formula \begin{equation}
\partial_{t}^{2}\log\E(e^{-\phi_{t}})=\mbox{Tr}\left(T[\partial_{t}\partial_{\bar{t}}\phi_{t}]+\left((T[\partial_{t}\phi_{t}])^{2}-T[(\partial_{t}\phi_{t})^{2}]\right)\right),\label{eq:sec deriv of log moment general}\end{equation}
 where $\mbox{Tr}$ denotes the trace and $T[f]$ is the Toeplitz
operator with symbol $f$ wrt the perturbed weight $\Phi+\phi_{t}:$
\[
T[f]=\int_{X}f(y)K_{\Phi+\phi_{t}}(\cdot,y)(=\Pi_{\Phi+\phi_{t}}(f\cdot))\]
 One then uses the geodesic equation \ref{eq:geod eq} to replace
$\partial_{\bar{t}}\partial_{t}\phi_{t}$ with $|\bar{\partial}_{X}(\partial_{t}\phi_{t})|_{\omega_{\phi_{t}}}^{2}$
in the first term in \ref{eq:sec deriv of log moment general} and
finally apply the Hörmander-Kodaira $L^{2}-$estimate for the inhomogenous
$\bar{\partial}_{X}$- equation (see \ref{eq:h-kodaria sharp case}
below) to deduce that $\partial_{t}^{2}\log\E(e^{-\phi_{t}})\leq0.$

\subsubsection{The proof of Proposition \ref{pro:moser for omega psh}}

Now consider the following functional on $\mathcal{C}^{0}(X)$, which
is invariant under addition of constants: \[
\mathcal{F}_{\omega}(\phi):=\mathcal{E}_{\omega_{0}}(\phi)+\frac{1}{N}\log\E(e^{-\phi})\]
For any given $\phi\in C^{0}(X)\cap\mbox{Psh}(X,\omega)$ we let $\phi_{t}$
be the $C^{0}-$geodesic such that $\phi_{0}=0$ and $\phi_{1}=\phi.$
By the concavity of $\mathcal{F}_{\omega}(\phi_{t})$ (resulting from
$(i')$ combined with $(ii')$ above) and since $\mathcal{F}_{\omega}(\phi_{0})=0$
we have \[
\mathcal{F}_{\omega}(\phi)\leq d(\mathcal{F}_{\omega}(\phi_{t}))/dt_{t=0}=\int(V\E(\delta/N)/\omega-1)\frac{1}{V}\omega(-d\phi_{t}/dt)_{t=0}\]
 Next, note that, since the inequality in the theorem that we are
about to prove is invariant under $\phi\rightarrow\phi+C$ we may
as well assume that $\sup_{X}\phi=0.$ Since $\phi_{t}$ is convex
in $t$ we have $-d\phi_{t}/dt\leq\phi_{1}-\phi_{0}=\phi$ (we are
using \emph{right }derivatives, which always exist by convexity) and
hence \[
\mathcal{F}_{\omega}(\phi)\leq\sup_{X}(V\E(\delta/N)/\omega-1)\frac{1}{V}(\int\omega(-d\phi_{t}/dt)_{t=0}\]
Next, note that, combining $(ii)$ and $(ii')$ above gives \[
(\int\omega(-d\phi_{t}/dt)_{t=0}=d\mathcal{E}_{\omega}(\phi_{t})/dt_{t=0}=-\mathcal{E}_{\omega}(\phi)\]
 and hence \[
\mathcal{F}_{\omega}(\phi)\leq\sup_{X}(V\E(\delta/N)/\omega-1)(-\mathcal{E}_{\omega}(\phi))\]
Finally, replacing $\phi$ with $\phi-\sup_{X}\phi$ finishes the
proof of the proposition.

\subsubsection{The psh projection $P_{\omega}$}

To reduce the case of a general smooth function $\phi$ to an $\omega-$psh
one we will make use of the psh-projection $P_{\omega}$ mapping smooth
functions to $\omega-$psh ones: \begin{equation}
(P_{\omega}\phi)(x):=\sup\left\{ \psi(x):\,\psi\in PSH(X,\omega),\,\psi\leq\phi\,\,\textrm{on\ensuremath{\, X}}\right\} \label{eq:def of proj as reg env}\end{equation}
It is not hard to see that $P_{\omega}\phi$ is continuous when $\phi$
is and moreover that the following {}``orthogonality relation''
holds \cite{b-b} \begin{equation}
\int_{X}(\phi-P_{\omega}\phi)dd^{c}(P_{\omega}\phi)=0\label{eq:og relation}\end{equation}
(as a consequence of the maximum principle for the Laplacian). 
\begin{prop}
Let $(X,\omega)$ be a Riemann surface with a Kähler. Then \[
(i)\,\mathcal{E_{\omega}}(\phi)\leq\mathcal{E_{\omega}}(P_{\omega}\phi),\,\,\,(ii)\,\left\Vert d(P_{\omega}\phi)\right\Vert _{X}^{2}\leq\left\Vert d\phi\right\Vert _{X}^{2}\]
 for any $\phi\in C^{\infty}(X).$ \end{prop}
\begin{proof}
$(i)$ was proved in \cite{berm5} and $(ii)$ is proved in a similar
way, as we will next see. Integrating by parts (which is allowed,
for example using that $P_{\omega}\phi$ is $\mathcal{C}^{1,1}-$
smooth \cite{berm5}) gives \[
\left\Vert d(P_{\omega}\phi)\right\Vert _{X}^{2}=\int(-P_{\omega}\phi)dd^{c}(P_{\omega}\phi)=\int(-P_{\omega}\phi)(dd^{c}P_{\omega}\phi+\omega)+\int(P_{\omega}\phi)\omega\]
Next, since $P_{\omega}\phi=\phi$ a.e. with respect to $(dd^{c}P_{\omega}\phi+\omega)$
(by formula \ref{eq:og relation}) this means that \[
\left\Vert d(P_{\omega}\phi)\right\Vert _{X}^{2}=\int(-\phi)(dd^{c}P_{\omega}\phi+\omega)+\int(P_{\omega}\phi)\omega=\int(-\phi)(dd^{c}P_{\omega}\phi)+\int(P_{\omega}\phi-\phi)\omega\]
But since $(P_{\omega}\phi-\phi)\leq0$ and $\omega\geq0$ the last
term above is non-positive and hence \[
\left\Vert d(P_{\omega}\phi)\right\Vert _{X}^{2}\leq\left\Vert d(P_{\omega}\phi)\right\Vert _{X}\left\Vert d\phi\right\Vert _{X},\]
also using the Cauchy-Schwartz inequality for the first term above.
Dividing out $\left\Vert d(P_{\omega}\phi)\right\Vert _{X}$ (which
is always non-zero if $\phi$ is) proves Step 2.
\end{proof}

\subsubsection{End of proof of Theorem \ref{thm:(determinatal-moser...)...For-the}}

We start with the proof of the inequality \ref{eq:thm det ineq}.
Consider the line bundle $kL$ with $\Phi$ the weight of a metric
on $L$ with curvature $\omega:=dd^{c}\Phi>0$ and decompose \[
kL=:L_{k}+K_{X},\,\,\, k\Phi=:\Phi_{k}+\Phi_{\omega}\]
 where $\Phi_{\omega}:=\log(\frac{\omega}{Vidz\wedge d\bar{z}})$
defines the weight of a metric on on $K_{X}.$ Then the Hilbert space
$H^{0}(kL)$ associated to the weighted measure $(\frac{\omega}{V},\omega)$
is naturally isomorphic to the Hilbert space $H^{0}(L_{k}+K_{X})$
associated to the weight $\Phi_{k}$ in the adjoint setting, just
using that, by definition, \[
e^{-k\Phi}\frac{\omega}{V}=e^{-\Phi_{k}}idz\wedge d\bar{z}\]
 We will write $\omega_{k}:=dd^{c}\Phi_{k}$ (and we let $N_{k}$
be the dimension of $H^{0}(kL)$ and $V_{k}$ the volume (degree)
of $L_{k}.$ Then \begin{equation}
\omega_{k}/V_{k}=\omega/V\label{eq:normaliz omega}\end{equation}
and in particular $\omega_{k}>0.$ This follows immediately from the
fact that the forms in rhs and the lhs above both integrate to one
over $X$ and moreover, by assumption, $\omega$ satisfies the Kähler-Einstein
equation: \[
dd^{c}\phi_{\omega}(:=\mbox{-Ric\ensuremath{\omega}) }=\lambda\omega\]
 for some constant $\lambda,$ so that $\omega_{k}$ is proportinal
to $\omega.$ 

\emph{Step one:} scaling by $k$ and assuming $(\omega_{k})_{\phi}(:=\omega_{k}+dd^{c}\phi)\geq0.$

Applying Prop \ref{pro:moser for omega psh} and Prop \ref{pro:bergm as}
to $(L_{k},\omega_{k})$ and $\phi$ and using formula \ref{eq:eomega as dirch}
gives, using \ref{eq:normaliz omega}, \[
\frac{1}{N_{k}}\log\E(e^{-\phi})+\mathcal{E}_{\omega_{k}}(\phi)\leq\epsilon_{k}\left(\frac{1}{2V_{k}}\left\Vert d\phi\right\Vert _{X}^{2}+\int(\sup_{X}\phi-\phi)\frac{\omega}{V}\right)\]
 Next, we recall the following basic inequality: there is a constant
$C$ (only depending on $\omega)$ such that \[
\sup_{X}\psi\leq\int_{X}\psi\omega+C\]
 for any $\psi$ such that $\omega_{\psi}\geq0$ (as follows immediately
from Green's formula; see \cite{g-z} for more general inequalities).
Setting $\psi=\phi/k$ and applying the previous inequality to the
rhs in the preceeding inequality gives, since $\omega_{k}/k\sim\omega,$
that \[
\frac{1}{N_{k}}\log\E(e^{-\phi})+\mathcal{E}_{\omega_{k}}(\phi)\leq\epsilon_{k}(\left\Vert d\phi\right\Vert _{X}^{2}+kC)\]

\emph{Step two:} using $P_{\omega_{k}}$

Let now $\phi$ be a general smooth function. Since $P_{(\omega_{k})}\phi\leq\phi$
we have $\frac{1}{N}\log\E(e^{-\phi})\leq\frac{1}{N}\log\E(e^{-P_{(\omega_{k})}\phi})$
and hence the previous step applied to $P_{\omega_{k}}\phi$ combined
with $(i)$ in the previous proposition and step one gives \[
\frac{1}{N_{k}}\log\E(e^{-\phi})+\mathcal{E}_{\omega_{k}}(\phi)-\epsilon_{k}\leq\epsilon_{k}\left\Vert d(P_{\omega}\phi\right\Vert _{X}^{2}\leq\epsilon_{k}\left\Vert d\phi\right\Vert _{X}^{2}\]
 also using $(ii)$ in the previous proposition in the last inequality.
Finally,using the scaling property \begin{equation}
\log\E(e^{-(\psi+c)})/N=-c+\log\E(e^{-\psi})/N\label{eq:scal prop of log moment}\end{equation}
together with formula \ref{eq:eomega as dirch} and the identity \ref{eq:normaliz omega}
we can rewrite \[
\frac{1}{N_{k}}\log\E(e^{-\phi})+\mathcal{E}_{\omega_{k}}(\phi)=\frac{1}{N_{k}}\log\E(e^{-(\phi-\int_{X}\phi\frac{\omega}{V}})-\frac{1}{V_{k}}\frac{1}{2}\left\Vert d\phi\right\Vert _{X}^{2}\]
All in all this means that \[
\log\E(e^{-\phi})\leq(\frac{N_{k}}{V_{k}}\frac{1}{2}+\epsilon_{k})\left\Vert d\phi\right\Vert _{X}^{2}+\epsilon_{k}\]
Finally, by the Riemann-Roch theorem 

\[
\frac{N_{k}}{V_{k}}=\frac{k\deg(L)-\deg(K_{X})/2}{k\deg(L)-\deg(K_{X})}=\frac{N_{k}}{N_{k}-\deg(K_{X})/2}=\frac{N_{k}}{N_{k}+(1-g)}\]
finishing the proof of the inequality \ref{eq:thm det ineq}. 

To prove the second inequality \ref{eq:det moser intro fluct version}
in the theorem we first note that \[
\int\phi(\omega/V-\E(\delta/N))\leq\epsilon_{N}\left\Vert \phi\right\Vert _{L^{1}(X)/\R}(:=\epsilon_{N}\inf_{c\in\R}\left\Vert \phi+c\right\Vert _{L^{1}(X)})\]
Indeed, the lhs above is invariant under the action of $\R,$ $\phi\rightarrow\phi+c,$
and hence the inequality follows immediately from Prop \ref{pro:bergm as}.
The inequality \ref{eq:thm det ineq} then follows immediately from
the fact that $\phi\rightarrow d\phi$ is invariant under the action
of $\R$ combined with the scaling property \ref{eq:scal prop of log moment}
(just take $\psi=\phi-\int\phi\omega$ and $c=\int\phi(\omega/V-\E(\delta/N))).$

\subsection{Proof of Cor \ref{cor:sharp tail intro}}

The proof is a standard application of Markov's inequality: for any
given $t>0$ we have \[
\mbox{Prob}{}\{Y>1\}=\mbox{Prob}{}\{e^{tY}>e^{t}\}\leq e^{-t}\E(e^{tY}),\]
 where in our case $Y=\frac{1}{N\epsilon}(\phi(x_{1})+...+\phi(x_{N}).$
By the previous theorem the rhs above is bounded by $e^{-t+ct^{2}/2}e^{\epsilon_{N}}$
for $c=(a_{N}+\epsilon_{N})\left\Vert d(\frac{1}{N\epsilon}\phi)\right\Vert ^{2}.$
Taking $t=1/c$ shows that the first factor may be estimated by $e^{-\frac{1}{2c}}$
which finishes the proof of the corollary.

\subsection{Proof of Theorem \ref{thm: szeg=0000F6type} (Sharp Szegö type limit
theorem)}

We will use the following notation for the fluctuation of the linear
statistic determined by a function $\phi$ on $X:$ 

\[
\tilde{\phi}:=\sum_{i=1}^{N}(\phi(x_{i})-\E(\phi(x_{i}))\]
We start by proving the following universal bound on the variance
for the canonical processes, which is of independent interest.
\begin{prop}
For any given function $\phi$ on $X$ the following upper bound on
the variance of the corresponding linear statistic holds: \[
\E(|\tilde{\phi}|^{2})/4\leq(1+\epsilon_{N})\left\Vert d\phi\right\Vert _{X}^{2}+\epsilon_{N}\left\Vert \phi\right\Vert _{L^{1}(X)/\R}^{2}\]
where $\epsilon_{N}$ denotes a sequence, independent of $\phi,$
tending to zero. In particular, if $\phi\in L^{1}(X)$ and $d\phi\in L^{2}(X)$
then the variance is uniformly bounded from above by a constant independent
of $N.$\end{prop}
\begin{proof}
We will denote by $\epsilon_{N}$ a sequence tending to zero, which
may change from line to line. By the second inequality in Theorem
\ref{thm:(determinatal-moser...)...For-the} we have \[
\E(e^{-t\tilde{\phi}})\leq e^{(1+\epsilon_{N}))\frac{1}{2}t^{2}\left\Vert d\phi\right\Vert _{X}^{2}+\epsilon_{N}t\left\Vert \phi\right\Vert _{L^{1}(X)/\R}}e^{\epsilon_{N}}\]
Using $2ab\leq a^{2}+b^{2}$ hence gives \[
\E(e^{-t\tilde{\phi}})\leq e^{\frac{1}{2}t^{2}f_{N}}e^{\epsilon_{N}},\,\, f_{N}=\left((1+\epsilon_{N})\left\Vert d\phi\right\Vert _{X}^{2}+\epsilon_{N}\left\Vert \phi\right\Vert _{L^{1}(X)/\R}^{2}+\epsilon_{N}\right)\]
Repeating the argument in the proof of \ref{cor:sharp tail intro}
(involving Markov's inequality) hence gives \[
\mbox{Prob}{}\{(\tilde{\phi}>\lambda\}\leq e^{-\lambda^{2}\frac{1}{2}\frac{1}{f_{N}}}e^{\epsilon_{N}}\]
Now using the push-forward formula for the integral in $\E(|\tilde{\phi}|^{2})$
we can write \[
\E(|\tilde{\phi}|^{2})=\int_{0}^{\infty}\mbox{(Prob}{}\{\tilde{\phi}^{2}>\lambda\})d(\lambda^{2})+\int_{0}^{\infty}\mbox{(Prob}{}\{(-\tilde{\phi})^{2}>\lambda\})d(\lambda^{2})\]
\[
\leq2\cdot2f_{N}e^{\epsilon_{N}}\]
where we used that $\int_{0}^{\infty}e^{-\frac{1}{2}\frac{1}{a}s}ds=2a$
in the last step, finishing the proof.
\end{proof}
As shown in \cite{berm2} (see also the Remark below) we have for
any fixed\emph{ smooth} function $\phi$ and $t\in\R$ \begin{equation}
\E(e^{it\tilde{\phi}})\rightarrow e^{-t^{2}\frac{1}{2}\int_{X}d\phi\wedge d^{c}\phi}\label{eq:pf of clt}\end{equation}
as $N\rightarrow\infty.$ Using the variance estimate above we can
extend the previous convergence to the case when we merely assume
that $\left\Vert d\phi\right\Vert _{X}<\infty$ (and hence $\left\Vert \phi\right\Vert _{L^{1}(X)}<\infty).$
To this end take a sequence $\phi_{j}\in\mathcal{C}^{\infty}(X)$
such that $\left\Vert d(\phi_{j}-\phi)\right\Vert _{X}\rightarrow0$
and $\left\Vert \phi_{j}-\phi\right\Vert _{L^{1}(X)}\rightarrow0.$
Since \[
|\E(e^{it\tilde{\phi}_{j}})-\E(e^{it\tilde{\phi}})|^{2}\leq\E(|\tilde{u}|^{2})\]
 for $u=\phi_{j}-\phi$ (just using $1-e^{is}\leq|s|)$ we deduce
that \[
|\E(e^{it\tilde{\phi}_{j}})-\E(e^{it\tilde{\phi}})|^{2}\leq C(\left\Vert d(\phi-\phi_{j})\right\Vert _{X}^{2}+\left\Vert \phi-\phi_{j}\right\Vert _{L^{1}(X)}^{2}\]
 for $N>>1$ and hence letting first $N$ and then $j$ tend to infinity
proves the convergence \ref{eq:pf of clt} in the non-smooth case
as well.

Next, we observe that the convergence \ref{eq:pf of clt} moreover
holds for any $t\in\C.$ Indeed, \[
f_{k}(t):=\E(e^{it\tilde{\phi}})\]
is a sequence of holomorphic functions on $\C$ such that for $t$
in a fixed compact subset $K$ of $\C$ \[
|f_{k}(t)|\leq\E(e^{-(Im(t))\tilde{\phi}})\leq C_{K}\]
using the second inequality in Theorem \ref{thm:(determinatal-moser...)...For-the}.
Since $f_{k}$ converges point-wise to the holomorphic function $f(t)=e^{-t^{2}\int d\phi\wedge d^{c}\phi}$
for $t\in\R$ it hence follows (e.g. by Vitali's theorem) that $f_{k}$
converges to $f$ everywhere on $\C.$ In other words we have now
proved Theorem \ref{thm: szeg=0000F6type} for the case of real and
imaginary $\phi.$ Finally, if $\phi$ is complex valued we consider
$\phi_{s}=u+sv$ where $\phi=\phi_{s}$ for $s=i.$ The previous convergence
shows that $\E(e^{-\tilde{\phi_{s}}})$ converges to an explicit holomorphic
function (as above) for $s\in\R.$ Moreover, since the upper bound
on $|f_{k}(s)|$ still holds (by the same argument) the previous argument
also shows that the convergence holds for any $s\in\C$ and in particular
for $s=i.$
\begin{rem}
\label{rem:lower bd}For completenes we briefly indicate a {}``self-contained''
proof of \ref{eq:pf of clt} in the case when $\phi$ is smooth. Since
we already have established the upper bound on log $\E(e^{t\tilde{\phi}})$
it will be enough to establish the lower bound (and as above we may
assume that $t$ is real). To this end we use that, at level $k,$
\[
\partial_{t}^{2}\log\E(e^{-t\tilde{\phi}})=\frac{1}{2}\int_{X\times X}|K_{k\Phi+t\phi}(x,y)|^{2}e^{-((k\Phi+t\phi)(x)+(k\Phi+t\phi)(y)}(\phi(x)-\phi(y))^{2}\]
(which follows from \ref{eq:sec deriv of log moment general} using
that the first term vanishes and by rewriting the second term). Next
we restrict the integration to $A_{k}:=\{d(x,y)\leq\log k/k^{1/2}\}\subset X\times X.$
Let $z$ denote local holomorphic coordinates centered at $x\in X$
and a trivialization of $L$ such that $\Phi(z)=|z|^{2}+O(|z|^{3}).$
Then it is well-known that \[
K_{k\Phi+t\phi}(x+z/k^{1/2},x+w/k^{1/2})=ke^{z\bar{w}}+o(1)\]
 uniformly in $k$ and $t$ (for $t=0$ this follows immediately from
Theorem \ref{thm:bergman kernel as} below and the general case is
obtained from the same proof since the perturbation $t\phi$ does
not effect the leading term). Finally integrating first over $y$
(or rather $w)$ and then over $y$ gives the lower bound $\left\Vert d\phi\right\Vert ^{2}.$
Since, the first derivative of $\log\E(e^{-t\tilde{\phi}})$ at $t=0$
vanishes, using $(i)$ in \ref{eq:first variational form}, this finishes
the proof of the lower bound (by general integration theory). 
\end{rem}

\subsection{\label{sub:A-brief-acount}A brief acount of the higher dimensionsional
case}

Let us now come back to the case when $X$ is $n-$dimensional and
fix a Kähler form $\omega$ on $X.$ In \cite{berm2} the analogue
of the convergence in Theorem \ref{thm: szeg=0000F6type} was shown
to hold as long as $\phi$ is smooth. More precisely, in the convergence
statement $\phi$ has to be replaced by $k^{-(n-1)/2}\phi$ and the
norm $\left\Vert d\phi\right\Vert _{X}^{2}$ by \[
\left\Vert d\phi\right\Vert _{(X,\omega)}^{2}=\int d\phi\wedge d^{c}\phi\wedge\frac{\omega^{n-1}}{(n-1)!}(=\int|\nabla\phi|^{2}dV)).\]
 However, when $n>1$ there are integrable functions $\phi$ with
$\int_{X}|\nabla\phi|^{2}\omega^{n}<\infty,$ but $\int e^{-\phi}dV=\infty$
(as is well-known in the context of Sobolev inequalities). As a consequence,
it is not hard to check that for such a function $\phi$ we have $\E(e^{-(\phi(x_{1})+\cdots)})=\infty$
and in particular the analogue of the convergence in Theorem \ref{thm: szeg=0000F6type}
cannot hold (after perhaps scaling $\phi)$. Moreover, the corresponding
analogue of the Moser-Trudinger inequality in Theorem \ref{thm:(determinatal-moser...)...For-the}
fails when $n>1$ (as is seen by approximating $\phi$ as above with
a monotone smooth sequence $\phi_{j}$). Explicit counter-examples
are obtained, already when $N=1,$ by letting $X=\P^{n}(\supseteq\C_{z}^{n})$
and $\omega$ be the standard $SU(n+1)-$invariant metric on $\P^{n}$
and taking $\phi_{j}(z):=m\log(\frac{1/j+|z|^{2}}{1+|z|^{2}}$) (for
a fixed $m\geq n)$ decreasing to $\phi(z).$ Note that $\phi_{j}$
is even $\omega-$psh. 

On the other hand, another variant of the determinantal Moser-Trudinger
inequality in Theorem \ref{thm:(determinatal-moser...)...For-the}
does hold in higher dimensions. More precisely, $\frac{1}{2}\left\Vert d\phi\right\Vert _{X}^{2}$
has to be replaced by Aubin's $J-$functional (which is comparable
to $\int d\phi\wedge d^{c}\phi\wedge(\omega_{\phi})^{n}.$ Moreover
$\phi$ has to be assumed $\omega-$psh (i.e. $\omega_{\phi}\geq0)$
(otherwise there are counter-examples, as explained in \cite{berm5})
When $X=\P^{n}$ (or more generally $X$ is a rational homogenous
manifold) the corresponding inequality is the content of Cor 2 in
\cite{berm5}, with vanishing error terms $\epsilon_{N}.$ More generally,
the arguments in Step one in the proof of Theorem \ref{thm:(determinatal-moser...)...For-the}
extend in a straight-forward manner to the higher-dimensional case
when the Kähler metric $\omega$ has a constant scalar curvature (but
then the error terms $\epsilon_{N_{k}}$ are then of the order $O(1/k)).$

\section{\label{sec:Bergman-kernel-asymptotics}Bergman kernel asymptotics
with expontially small error terms }

In this section we will prove the following theorem used in the proof
of Proposition \ref{pro:bergm as} above (for an explicit description
of $\delta$ below, see Remark \ref{rem:explicit delta}).
\begin{thm}
\label{thm:bergman kernel as}Let $L\rightarrow X$ be a line bundle
over a Riemann surface equipped with a metric $e^{-\Phi}$ with positive
curvature form $\omega(=dd^{c}\Phi)$ such that the Riemannian metric
on $X$ defined by $\omega$ has constant scalar curvature $R$ close
to $x.$ Then there is a neighbourhood of $\{x\}\times\{x\}$ in $X\times X$
such that the corresponding Bergman kernel $K_{k}$ satisfies \begin{equation}
K_{k}(z,w)=(k+\frac{1}{2}R)e^{k\psi(\bar{z},w)}+\epsilon_{k}\label{eq:bergman kernel off-diag in theorem exp}\end{equation}
 where $\psi$ is the local holomorphic function such that $\psi(\bar{z},w)=\Phi(z)$
and $\epsilon_{k}$ denotes a smooth section of $kL\boxtimes kL$
whose point-wise norm is of the order $\mathcal{O}(e^{-\delta k}).$
In particular,

\begin{equation}
B_{k}(x):=\left\Vert K_{k}(x,x)\right\Vert =k+\frac{1}{2}R+\mathcal{O}(e^{-\delta k})\label{eq:bergman function in theorem exp small error}\end{equation}

\end{thm}
(in the case when $X$ is the two-sphere and $R$ is constant on all
of $X$ our arguments will give the well-known fact that the error
terms vanish identically). Here the Berman kernel $K_{k}\in H^{0}(X\times\bar{X,}kL\boxtimes kL)$
denotes the integral kernel of the orthogonal projection from $\mathcal{C}^{\infty}(X,L)$
onto the Hibert space $H^{0}(X,kL)$ using the $L^{2}-$norm defined
by the metric on $L$ and volume form $d\nu=\omega$ (formula \ref{eq:general l2 norm intro}).
The normalization of $R$ has been chosen so that $R=\deg(TX)(=2g-2)$
when it is globally constant (and hence integrating \ref{eq:bergman function in theorem exp small error}
against $\omega$ over $X$ gives the Riemann-Roch relation \ref{eq:r-r relation intro}
for $k$ large). We recall that a sequence $a_{k}$ is said to be\emph{
exponentially small}, written as $a_{k}=\mathcal{O}(e^{-k\delta})$
if $|a_{k}|\leq Ce^{-k\delta}$ for some numbers $C,\delta>0$ (if
$a_{k}$ are functions then, by definition, the estimate holds uniformly).

The case of larger error terms of the form $\mathcal{O}(e^{-(\log k)^{2}\delta})$
in \ref{eq:bergman function in theorem exp small error} was priouvsly
obtained in \cite{liu,liu-l} using Tian's method of peak sections.
It was also pointed out there that the case of even larger error terms
of the form $\mathcal{O}(k^{-\infty})$ can be deduced from the results
in \cite{l-t} concerning the Yau-Tian-Zelditch-Catlin expansion of
$B_{k},$ but that one may expect exponentially small error terms
(as confirmed in the theorem above). In the case when $L=K_{X}$ and
the scalar curvature is constant on\emph{ all }of $X$ (in particular
$X$ then has genus at least two and $R<0$) the error term $\mathcal{O}(e^{-\delta k})$
in \ref{eq:bergman function in theorem exp small error} could also
be obtained by writing $X=\Gamma/\H$ for a Fuchsian group $\Gamma$
and using that that $B_{k}$ is constant in the non-compact case setting
of $X=\H$ and then estimate the effect of the {}``$\Gamma-$periodization''
coming from a Poincaré theta series (as pointed out to the author
by Steve Zelditch) A similar periodization argument was used in \cite{fau}
in the case when $X$ is the torus. 

One motivation to consider the situation when $R$ is not globally
constant is to allow applications to the setting of constant curvature
metrics with conical singularities and cusps. For example, in the
hyperbolic setting this means that the Kähler form $\omega$ is the
unique solution to \begin{equation}
\mbox{Ric \ensuremath{\omega=-\omega+\sum_{i}c_{i}\delta_{P_{i}}}}\label{eq:sing einstein eq}\end{equation}
for given coefficents $c_{i}\in[0,1]\cap\Q$ and a finite number of
points $P_{i}$ in $X$ (\cite{heins}, Thm 21.1). Equivalently, $\omega$
has constant scalar curvature $-1$ on $X-\{P_{i}\}$ with \emph{conical
singularities} at an angle $2\pi(1-c_{i})$ at any $P_{i}$ such that
$c_{i}<1$ and a\emph{ cusp} at any $P_{i}$ such that $c_{i}=1.$
\footnote{The classical case when $X-\{P_{i}\}=\Gamma/\H$ for a Fuchsian group
$\Gamma$ corresponds to the the case when $c_{i}=1-1/m$ for $m$
a positive integer or infinity and then $\omega$ is induced from
the hyperbolic metric on $\H$ %
}

Letting $D=\sum_{i}c_{i}\mathcal{O}_{P_{i}}$ be the corresponding
$\Q-$ line bundle we then have the following
\begin{cor}
Let $L=K_{X}+D$ and let $\omega$ be the unique (singular) metric
on $X$ above. Then the Bergman kernel expansions \ref{eq:bergman kernel off-diag in theorem exp}
and \ref{eq:bergman function in theorem exp small error} hold for
any $x\in X-\{P_{i}\}.$ Moreover, the positive number $\delta$ appearing
in \ref{eq:bergman function in theorem exp small error} may be taken
to be arbitrarily close to \[
2\log(\cosh(\pi I(x)/\sqrt{2}))\]
 where $I(x)$ is the injectivity radius in $X-\{P_{i}\}$ at $x$
(which coincides with half the length of the shortest closed and simple
geodesic on $X,$ passing through $x,$ when $D=0).$ 
\end{cor}
The Bergman kernel in the previous corollary is, as usual, defined
wrt the subspace of $H^{0}(X,kL)$ consisting of all {}``cusp forms'',
i.e. sections vanishing at the cusps and it is well-defined for all
$k$ such that $kc_{i}\in\Z$ for all $i.$ 

The rest of the section is devoted to the proof of the theorem above;
following the scheme in \cite{berm0} we first prove a local variant
of the expansion and then globalize. The main point here is the observation
that the local expansion may be obtained using the {}``local symmetry''
of $L$ (as opposed to the general case treated in \cite{berm0})
which leads to a precise controle of the error terms.

As is well-known the local constant curvature condition implies that
there exists a local holomorphic coordinate $w$ centered at $x$
on some simply connected neighbourhood $U$ such that \[
\omega:=\frac{i}{2\pi}2(1+R|w|^{2})^{-2}dw\wedge d\bar{w}(=dd^{c}\Phi_{0})\]
 where $\Phi_{0}(w)=R^{-1}\log(1+R|w|^{2})$ for $R\neq0$ and $\Phi_{0}=|w|^{2}$
for $R=0$ (obtained in the limit $R\rightarrow0).$ Now fix a local
holomorphic section $s$ of $L$ close to $x$ and write $\left\Vert s\right\Vert ^{2}=e^{-\Phi}$
for a local function $\Phi,$ recalling that $\omega=dd^{c}\Phi$
(see section \ref{sub:Notation-and-general}). By the previous relation
this means that $\Phi-\Phi_{0}$ is a harmonic function on $U$ and
hence we may write $e^{-\Phi}=|h|^{2}e^{-\Phi_{0}}$ for some non-vanishing
holomorphic function $h$ on $U.$ Accordingly, after replacing $s$
with $h^{-1}s$ we may as well assume that we are in the\emph{ model
case} $\Phi=\Phi_{0}.$

\subsection{Local Bergman kernels for the model cases}

Given a smooth function $\Phi$ on a domain $U$ in $\C$ containing
$0$ we let \[
\left\langle f,g\right\rangle _{U,k\Phi}:=\int_{U}f\bar{g}e^{-k\Phi}dd^{c}\Phi\]
and denote by $H_{k\Phi}(U)$ the space of all holomorphic functions
on $U$ such $\left\Vert f\right\Vert _{U,k\Phi}^{2}(:=\left\langle f,g\right\rangle _{U,k\Phi})<\infty.$
Following \cite{berm0} we will say that \emph{$K_{(k)}(z,\zeta)$
is a (local) Bergman kernel mod $\mathcal{O}(e^{-k\delta})$ (with
respect to $\Phi)$ }if it is holomorphic in $\zeta$ and there exists
number $\delta>0$ such that for any $f\in H_{k\Phi}(U)$ we have,
for all $z$ in some neighbourhood $V\subset U$ of $0$ that \begin{equation}
f_{k}(z)=\left\langle f_{k},\chi K_{(k)}(z,\cdot)\right\rangle _{U,k\Phi}+\left\Vert f\right\Vert _{U,k\Phi}\mathcal{O}(e^{-k\delta})e^{k\Phi/2}\label{eq:def form for local bergm}\end{equation}
 where $\chi$ denotes a smooth function $\chi$ compactly supported
on $U$ which is equal to one on $\frac{1}{2}U.$ 
\begin{prop}
Let $\Phi(w)=-2\log(1+R|w|^{2})/R.$ Then the function $K_{(k)}(z,\zeta)=(k+\frac{R}{2})(1+R\zeta\bar{z})^{2k/R}$
is a local Bergman kernel mod $\mathcal{O}(e^{-k\delta})$ (wrt $\Phi)$
when $R\neq0$ and $K_{(k)}(z,\zeta)=ke^{\bar{z}\zeta}$ when $R=0$
(coinciding with the limit when $R\rightarrow0).$ \end{prop}
\begin{proof}
Let $\epsilon$ be a fixed (small) positive number. It will be convenient
to let $\delta$ be a small positive number (depending on $\epsilon)$
whose value may change from line to line. First we not that for any
$f\in H_{k\Phi}(U)$ we have \begin{equation}
f_{k}(0)=(k+\frac{R}{2})\int_{|w|<\epsilon}f_{k}e^{-k\Phi}dd^{c}\Phi+\left\Vert f_{k}\right\Vert _{\frac{1}{2}U,k\Phi}\mathcal{O}(e^{-k\delta})e^{k\Phi/2},\label{eq:local bergman at zero}\end{equation}
Indeed, applying the mean-value property of holomorphic functions
to $f_{k}(re^{i\theta})$ for $r$ fixed and then integrating over
$r$ (using that $\Phi$ only depends on $r)$ gives \[
f_{k}(0)=c_{k}\int_{|w|<\epsilon}fe^{-k\Phi}dd^{c}\Phi\]
where \begin{equation}
1/c_{k}=\int_{|w|<\epsilon}e^{-k\Phi}dd^{c}\Phi=\int_{|w|<C}e^{-k\Phi}dd^{c}\Phi-\int_{\epsilon<|w|<C}e^{-k\Phi}dd^{c}\Phi\label{eq:invers of c}\end{equation}
and where we take $C^{2}=-1/R$ when $R<0$ and $C=\infty$ otherwise.
Since, with $s=r^{2},$ \[
e^{-k\Phi}dd^{c}\Phi=\frac{2}{\pi}(1+Rr^{2})^{-2k/R-2}\frac{1}{2}d(r^{2})d\theta=\frac{d}{ds}\left(\frac{1}{k+R/2}(1+Rs)^{-2k/R-1}\right)ds\wedge\frac{d\theta}{2\pi}\]
the first integral in \ref{eq:invers of c} equals $1/(k+R/2)$ and
\begin{equation}
\int_{\epsilon<|w|<C}e^{-k\Phi}dd^{c}\Phi=\mathcal{O}(e^{-k\delta}),\,\,\,\delta=\Phi(\epsilon^{2})\label{eq:exp decay of integral}\end{equation}
The formula \ref{eq:local bergman at zero} then follows from the
trivial relation $(1+\mathcal{O}(e^{-k\delta}))^{-1}=1+\mathcal{O}(e^{-k\delta})$
combined with the Cauchy-Schwartz inequality.

Next, we fix $z$ in a given (small) neighbourhood $V$ of $U$ and
define \begin{equation}
F_{z}(w):=\zeta:=(z-w)/(1+R\bar{z}w)\label{eq:def of map f}\end{equation}
(which is invertible with $w=F_{w}(\zeta))$ mapping $0$ to $z$
and \[
g_{z}(w):=f_{k}(\zeta)e^{k\psi(\bar{z},w)},\,\,\,\psi(z,w):=R^{-1}\log(1+Rzw)\]
for a given $f_{k}\in H_{k\Phi}(U).$ Then \begin{equation}
|g_{z}(w)|^{2}e^{-k\Phi(w)}e^{-k\Phi(z)}=|f_{k}(\zeta)|^{2}e^{-k\Phi(\zeta)}\label{eq:transf or ptwise norms in disc}\end{equation}
as follows immediately from the relation \begin{equation}
\psi(\bar{z},w)+\psi(z,\bar{w})-\Phi(w)-\Phi(z)=-\Phi(\zeta)\label{eq:first relation}\end{equation}
(see section \ref{sub:Proofs-of-the} below). This shows in particular
that $g_{z}\in H_{k\Phi}(U).$ Applying the formula \ref{eq:local bergman at zero}
to $f_{k}:=g_{z}$ hence gives \[
f_{k}(z)=(k+\frac{R}{2})\int_{|w|<\epsilon}f_{k}(\zeta)e^{k\psi(\bar{z,}w)}{}^{-k\Phi(w)}d_{w}d_{w}^{c}\Phi+\left\Vert g_{z}\right\Vert _{\frac{1}{2}U,k\Phi}\mathcal{O}(e^{-k\delta}),\]
To rewrite this we first note that $dd^{c}\Phi$ is invariant under
the map $F_{z}$ (as follows immediately from differentiating \ref{eq:first relation})
and hence \ref{eq:transf or ptwise norms in disc} gives that \[
e^{-k\Phi(z)}\int_{w\in\frac{1}{2}U}|g_{z}(w)|^{2}e^{-k\Phi(w)}d_{w}d_{w}^{c}\Phi=\int_{w\in\frac{1}{2}U}|f_{k}(\zeta)|^{2}e^{-k\Phi(\zeta)}d_{\zeta}d_{\zeta}^{c}\Phi,\]
 i.e. that \[
\left\Vert g_{z}\right\Vert _{\frac{1}{2}U,k\Phi}e^{-k\Phi(z)/2}=\left\Vert f\right\Vert _{F_{z}(\frac{1}{2}U),k\Phi}(\leq\left\Vert f\right\Vert _{U,k\Phi})\]
Moreover, the relation \begin{equation}
\psi(\bar{z},w)-\Phi(w)=\psi(z,\bar{\zeta})-\Phi(\zeta)\label{eq:second relation}\end{equation}
(see section \ref{sub:Proofs-of-the} below) then gives that \begin{equation}
f_{k}(z)=(k+\frac{R}{2})\int_{F_{z}(\epsilon D)}f_{k}(\zeta)e^{k\psi(z,\bar{\zeta})-k\Phi(\zeta)}d_{\zeta}d_{\zeta}^{c}\Phi+\left\Vert f\right\Vert _{U,k\Phi}e^{k\Phi(z)/2}\mathcal{O}(e^{-k\delta}).\label{eq:f as int}\end{equation}
Next, we note that by the Cauchy-Schwartz inequality (applied to $f_{k}$
and $e^{k\psi})$ and the relation \ref{eq:first relation} (applied
to $w=\zeta)$ we have \[
|\int f_{k}(\zeta)e^{k\psi(z,\bar{\zeta})-k\Phi(\zeta)}d_{\zeta}d_{\zeta}^{c}\Phi|^{2}e^{-\Phi(z)}\leq|\int|f_{k}(\zeta)|^{2}e^{-k\Phi(\zeta)}d_{\zeta}d_{\zeta}^{c}\Phi|||\int e^{-k\Phi(w)}d_{w}d_{w}^{c}\Phi|^{2}\]
By \ref{eq:exp decay of integral} the second factor in the rhs above
is exponentially small when integrating over the complement of a small
disc centered at $w=0$ (i.e. a small neighbourhood of $z$ in the
$\zeta$ coordinates). Hence, we may as well replace the integration
region $F_{z}(\epsilon D)$ in \ref{eq:f as int} with all of $U$
at the expense of introducing the cut-off function $\chi,$ concluding
the proof of the proposition.
\end{proof}

\subsubsection{\label{sub:Proofs-of-the}Proofs of the relations \ref{eq:first relation}
and \ref{eq:second relation} by lifting}

The relations \ref{eq:first relation} and \ref{eq:second relation}
are without doubt well-known (and trivial for $R=0)$, but for completeness
we give a brief proof here. To this end we use a standard lifting
argument. Geometrically, this amounts to lifting $F_{z}$ above to
an isometry of line bundles: $L\otimes\bar{L}_{0}\rightarrow L\otimes\bar{L}_{z}.$
Consider the vector space $\C^{2}$ equipped with the diagonal Hermitian
bi-linear form with eigenvalues $(R,1)$ so that the corresponding
squared pseudo-norm $\left\Vert \cdot\right\Vert _{R}^{2}$ is given
by $\left\Vert (w_{1},w_{2})\right\Vert _{R}^{2}=R|w_{1}|^{2}+|w_{2}|^{2}.$
Let \[
M_{z}:=\left(\begin{array}{cc}
-1 & z\\
R\bar{z} & 1\end{array}\right),\,\,\,\mbox{\ensuremath{\pi}}(w_{1},w_{2})=w_{1}/w_{2}\]
(assuming $w_{2}\neq0),$ where clearly $M_{z}$ preserves $\left\Vert \cdot\right\Vert _{R}^{2}$
modulo the scaling factor $\det(M_{z})=\left\Vert (z,1)\right\Vert _{R}^{2}.$
In particular, $\left\Vert M_{z}(w,1)\right\Vert _{R}^{2}=\left\Vert (z,1)\right\Vert _{R}^{2}\left\Vert (w,1)\right\Vert _{R}^{2}$
and since $\zeta(:=F_{z}(w))=\pi(M_{R}(w,1))$ this proves (upon taking
logarithms) the relation \ref{eq:first relation}. The relation \ref{eq:second relation}
now follows by substituting the relation $(1+R|z|^{2})=(1+R\bar{z}w)(1+R\bar{z}\zeta$)
into \ref{eq:first relation}. In turn, this latter relation can be
obtained by first calculating $d\zeta/dw=-(1+R|z|^{2})/(1+R\bar{z}w)^{2}$
and similarly for $\zeta$ replaced with $w$ (using that $w=F_{z}(\zeta)).$
Since $d\zeta/dw=(dw/d\zeta)^{-1}$ this forces the previous relation,
finishing the proof of \ref{eq:second relation}.

\subsection{Globalization}
\begin{prop}
Let $L\rightarrow X$ be a positive Hermitian holomorphic line bundle
over a compact complex manifold $X$ and let $x\in X$ be a fixed
point such that the local weight $\Phi$ of the metric wrt some trivialization
of $L$ around $x$ is real-analytic ad admits a local Bergman kernel
$K_{(k)}$ mod $\mathcal{O}(e^{-k\delta})$ such that \begin{equation}
K_{(k)}(z,\zeta)=a_{k}e^{k\psi(z,\bar{\zeta})}\label{eq:statement of prop globalization}\end{equation}
 for some sequence $a_{k}$ with sub-exponential growth (i.e. $|a_{k}|\leq C_{\delta}e^{k\delta}$
for any $\delta>0).$ Then the (global) Bergman kernel $K_{k}$ associated
to $kL$ satisfies the uniform estimate \[
\left\Vert K_{k}-K_{(k)}\right\Vert _{k\Phi}\leq Ce^{-\delta k}\]
on some neighbourhood $U\times U$ of $\{x\}\times\{x\}$ for some
numbers $C,\delta>0.$\end{prop}
\begin{proof}
Take local holomorphic coordinates $w$ centered at $x.$ The proof
of the proposition is essentially contained in the globalization argument
used in \cite{berm0}. For completeness we recall the argument. Fixing
$z$ and applying the defining formula \ref{eq:def form for local bergm}
$u_{k}:=K_{k,z}:=K_{k}(z,\cdot)$ gives \[
K_{k,z}=\left\langle \chi K_{k,z},K_{(k)}\right\rangle _{U,k\Phi}+\mathcal{O}(e^{-k\delta})e^{\Phi(z)/2}\]
 where we have used that $\left\Vert K_{k,z}\right\Vert _{U,k\Phi}^{2}\leq\left\Vert K_{k,z}\right\Vert _{X,k\Phi}=K_{k}(z,z)\leq Ck^{n}e^{k\Phi(z)}$
by a standard estimate for Bergman functions (as can be see from a
simple argument using the mean value property of holomorphic functions,
just as below). Next, we note that the difference $u_{k,z}:=K_{(k),z}-\left\langle \chi K_{k,z},K_{(k)}\right\rangle _{U,k\Phi}$
is the $L^{2}-$minimal solution to the $\bar{\partial}-$equation\begin{equation}
\bar{\partial}u=g,\label{eq:inhomo dbar eq}\end{equation}
with $g=\bar{\partial}(\chi K_{(k),z}),$ which by the Hörmander-Kodaira
$L^{2}-$estimate satisfies \begin{equation}
\left\Vert u_{k,z}\right\Vert _{k\Phi}^{2}\leq C\left\Vert g\right\Vert _{k\Phi}^{2}=\left\Vert (\bar{\partial}\chi)K_{(k),z}\right\Vert _{k\Phi}^{2},\label{eq:l2 estimate}\end{equation}
 (recall that $\bar{\partial}\chi$ is supported in a neighbourhood
of $w=0,$ vanishing close to $w=0).$ Hence the assumption \ref{eq:statement of prop globalization}
combined with the general basic fact that $e^{k2\mbox{Re}\psi(z,\bar{w})-\phi(w)-\phi(z)}$
is exponentially concentrated around $w=z$ (when $dd^{c}\phi>0)$
show that \[
\left\Vert u_{k,z}\right\Vert _{k\Phi,X}^{2}e^{-k\Phi(z)}\leq Ce^{-\delta k}.\]
 (in the case of Theorem \ref{thm:bergman kernel as} we get the same
$\delta$ as in \ref{eq:exp decay of integral} using \ref{eq:first relation}
as above). It is now a standard matter to convert this $L^{2}-$estimate
to an $L^{\infty}-$estimate for $|u_{k,z}(\zeta)|^{2}$ when $\zeta$
is close to $z.$ Indeed restricting the integration in the previous
inequality to a small disc $D_{k}(\zeta)$ of radius $\epsilon k^{-1/2}$
centered at $\zeta$ gives \[
\int_{D_{k}(\zeta)}|u_{k,z}(w)|^{2}dw\wedge d\bar{w}\leq C'e^{-\delta k}.\]
 Finally, since the integral in the lhs above may, by the mean value
property of holomorphic functions, be estimated from below by $c|u_{k,z}(\zeta)|^{2}/k$
this finishes the proof of the proposition.\end{proof}
\begin{rem}
\label{rem:explicit delta}Tracing through the arguments above in
fact gives an explicit expression for the exponent $\delta$ appearing
in the error terms in Theorem \ref{thm:bergman kernel as}. Indeed,
if we take $U$ as a disc $D_{r}$ of radius $r$ then $\delta$ can
be taken to be arbitrary close to $\Phi(r^{2}).$ To see this just
let the cut-off function $\chi$ instead be supported on $(1-\epsilon')U$
for a given $\epsilon'.$ Then $\epsilon$ appearing in \ref{eq:exp decay of integral}
can be taken arbitrarily close to $r(1-\epsilon').$ Note that if
$l$ is the radius of $D_{r}$ in the metric $\omega$ then, if $R$
is globally constant, the optimal choice of $r$ above correponds
to $l$ beeing the injectivity radius of $X$ at $x.$ A direct computation
gives $r=\frac{1}{\sqrt{-R}}\tanh(\sqrt{-\frac{\pi R}{2}}l)$ and
hence $\Phi(r^{2})=\frac{2}{-R}\log(\cosh(\sqrt{-\frac{\pi R}{2}}l)$.
Moreover, the proof of Theorem \ref{thm:bergman kernel as} also goes
through, word for word, in any dimension $n$ (so that $w=(w_{1},...,w_{n})$
etc) if one assumes that the Kähler metric $\omega$ has \emph{constant
holomorphic sectional curvature} (and in particular constant scalar
curvature $:=R).$ Indeed, using the normal coordinates in \cite{bo}
one reduces to $\Phi(z)=-(n+1)\log(1+R|w|^{2})$ as before. Computing
the integrals in \ref{eq:invers of c} then shows that the $k-$dependent
leading constant $k+R/2$ in the theorem has to be replaced by a constant
which is an explicit polynomial in $k$ which may be expanded as $k^{n}+\frac{R}{2}+O(k^{n-1}).$
Also note that the \emph{global} positivity (i.e. on $X-U)$ of the
line bundle in the previous proposition was only used in the Hörmander-Kodaira
$L^{2}-$estimate \label{rem:Tracing-through-the} (and that $C=C_{k}$
has sub-exponential growth in $k)$. For example, it holds as long
as $L$ is ample and globally \emph{semi}-positively curved (and positively
curved on $U)$. 
\end{rem}

\subsection{Proof of the Corollary}

The cororally follows from Theorem \ref{thm:bergman kernel as} and
the remark above. Indeed, writing $\omega=dd^{c}\Phi$ we have by
assumption that $dd^{c}\Phi=ce^{\Phi-\Phi_{D}}dz\wedge d\bar{z}$
where $\Phi_{D}=\log|s_{D}|^{2}$ for $s_{D}$ a holomorphic (multi-)section
of $D.$ Hence, $e^{-k\Phi}dd^{c}\Phi=e^{-\Phi_{k}},$ where $\Phi_{k}=(k-1)\Phi+\Phi_{D}$
is the weight of a singular metric on $L_{k}:=(k-1)L+D$ (i.e. $kL=L_{k}+K_{X})$
with positive curvature current. It then follows from Demailly's singular
version of the Hörmander-Kodaira $L^{2}-$estimates \cite{de4} that
the $L^{2}-$minimal solution $u$ to \ref{eq:inhomo dbar eq} satisfies
\begin{equation}
\left(\int_{X}|u|^{2}e^{-k\Phi}dd^{c}\Phi\right)=\int_{X}|u|^{2}e^{-\Phi_{k}}\leq\int_{X}\frac{g\bar{g}}{dd^{c}\Phi_{k}}e^{-\Phi_{k}}\label{eq:h-kodaria sharp case}\end{equation}
where $g$ is seen as a $(0,1)-$form with values in $L_{k}+K_{X}.$
In our case we take as above $g=\bar{\partial}(\chi K_{(k),z})$ which
is supported where $\Phi$ is smooth and hence the previous estimates
go through word for word (with a constant $C$ only depending on the
fixed point $x).$ Finally, since the space of all {}``cusp sections''
of $kL$ coincides with the subspace of $H^{0}(X,kL)$ of all sections
which are in $L^{2}$ wrt the $L^{2}-$norm defined by $\Phi_{k}$
this finishes the proof. This last fact follows from the well-known
fact \cite{heins} that $\Phi$ has only a mild singularity at any
cusp (corresponding to $z=0):$ $\Phi\sim-\log(-(\log|z|)).$ Hence,
since $dd^{c}\Phi=ce^{\Phi}\frac{1}{|z|^{2}}dz\wedge d\bar{z},$ a
local holomorphic function $u_{k}$ is locally integrable square wrt
$e^{-k\Phi}dd^{c}\Phi$ iff $u_{k}(0)=0.$

\end{document}